\documentclass[runningheads]{llncs}

\usepackage{amsmath,amssymb}
\usepackage{microtype}
\usepackage{color}
\usepackage{stmaryrd}
\usepackage{proof}
\usepackage[vlined,linesnumbered,algoruled]{algorithm2e}
\usepackage[hide]{ed}
\usepackage{hyperref}
\usepackage{enumitem}

\usepackage{tikz}
\usetikzlibrary{arrows}
\usetikzlibrary{decorations.pathmorphing}

\newcommand{\nc}{\newcommand}
\nc{\4}{\mathsf{4}}
\nc{\Cg}{\mathsf{Cg}}
\nc{\cless}{\preccurlyeq}
\nc{\Con}{\mathsf{Con}}
\nc{\Conl}{\mathsf{Con}_L}
\nc{\Conr}{\mathsf{Con}_R}
\nc{\Cut}{\mathsf{Cut}}
\nc{\D}{\mathsf{D}}
\nc{\defs}{:=}
\nc{\deriv}{\vdash}
\nc{\exptime}{\mathsf{EXPTIME}}
\nc{\forces}{\Vdash}
\nc{\G}{\mathsf{G}}
\nc{\GDL}{{\G_{\bMDL}}}
\nc{\harm}{\mathtt{hrm}}
\nc{\harmenemy}{\mathtt{hrm\_en}}
\nc{\K}{\mathsf{K}}
\nc{\Land}{\bigwedge}
\nc{\last}[1]{\mathsf{last}\!\left(#1\right)}
\nc{\lastl}[1]{\mathsf{last}_L\!\left(#1\right)}
\nc{\lastr}[1]{\mathsf{last}_R\!\left(#1\right)}
\nc{\Lor}{\bigvee}
\nc{\lrarr}{\leftrightarrow}
\nc{\MDL}{\mathsf{bMDL}}
\nc{\Mimamsa}{{M\={\i}m\={a}\d{m}s\={a}}}
\nc{\Mon}{\mathsf{Mon}}
\nc{\NEC}{\Box}
\nc{\Obl}{\mathcal{O}}
\nc{\Ob}{\Obl}
\nc{\plusplus}{{+\!\!+}}
\nc{\Pow}{\mathcal{P}}
\nc{\rarr}{\rightarrow}
\nc{\seq}{\Rightarrow}
\nc{\sts}[1]{| #1 |}
\nc{\syena}{\mathtt{sy}}
\nc{\T}{\mathsf{T}}
\nc{\trs}[1]{\llbracket #1 \rrbracket}
\nc{\Var}{\mathsf{Var}}
\nc{\W}{\mathsf{W}}

\nc{\bMDL}{\mathsf{bMDL}}
\newcommand{\Sf}{\mathsf{S4}}
\newcommand{\Sfi}{\mathsf{S5}}

\newcommand{\desharm}{\mathtt{des\_hrm\_en}}

\title{{\Mimamsa} deontic logic: proof theory and
  applications\thanks{Supported by FWF START project Y544-N23, FWF project V400 and EU
    \mbox{H2020-MSCA} grant 660047.}}

\author{Agata Ciabattoni\inst{1} \and Elisa Freschi\inst{2} \and Francesco A. Genco\inst{1} \and Bj{\"o}rn Lellmann\inst{1}}
\institute{Vienna University of Technology \\ \email{agata@logic.at genco@logic.at lellmann@logic.at} \and Institute for the Cultural and Intellectual History of Asia,\\Austrian Academy of Sciences, Vienna \\ \email{elisa.freschi@gmail.com}}

\date{}
\begin{document}
\maketitle
\date{}

\begin{abstract}
Starting with the deontic principles in {\Mimamsa} texts we introduce a new deontic logic. 
We use general proof-theoretic methods to obtain a cut-free sequent calculus for this logic, resulting in decidability, complexity results and neighbourhood semantics. 
The latter is used to analyse a well known example of conflicting obligations from the Vedas.

\end{abstract}

\section{Introduction}

\label{sec:introduction}
We provide a first bridge between formal logic and the {\Mimamsa} school of Indian philosophy.
Flourishing between the last centuries BCE and the 20th century, the main focus of this school is 
the interpretation of the prescriptive part of the Indian Sacred Texts (the \emph{Vedas}).
In order to explain ``what has to be done'' according to the Vedas,
{\Mimamsa} authors have proposed a rich body of 
deontic, hermeneutical and linguistic principles (\textit{metarules}), called \emph{ny\={a}ya}s, 
which were also used to find rational explanations for seemingly contradicting obligations.

Even though the {\Mimamsa} interpretation of the Vedas has pervaded almost every other school of 
Indian philosophy, theology and law, little research has been done on
the \emph{ny\={a}ya}s. Moreover, since not many scholars working on
{\Mimamsa} are trained in formal logic, and the untranslated texts are
inaccessible to logicians,
these deontic 
principles have not yet been studied
using methods 
of formal logic.  

In this paper starting from 
the deontic \emph{ny\={a}ya}s we define a new logic -- \emph{basic {\Mimamsa} 
deontic logic} ($\bMDL$ for short) -- that simulates {\Mimamsa} reasoning. 
After introducing the logic as an extension of modal logic
$\mathsf{S}\4$ with axioms obtained by formalising these principles
\footnote{While some of the
\textit{ny\={a}ya}s we consider are listed in the Appendix of \cite{Kane}, we extracted the remaining ones
directly from {\Mimamsa} texts, see \cite{Atinerpaper}.} and providing
a cut-free sequent calculus and neighbourhood-style semantics for it, we use
$\bMDL$
to 
reason about a well known example of seemingly conflicting
obligations contained in the Vedas. This example concerns
the malefic sacrifice called \emph{\'Syena} 
and proved to be a stumbling block for many {\Mimamsa} scholars.
The solution to this controversy provided by the 
semantics of $\bMDL$ turns out to coincide with
that of Prabh\={a}kara, one of the chief {\Mimamsa} authors, which
previous 
approaches failed to make sense of, e.g., \cite{Stcherbatsky1926}.
Our formal analysis relies essentially on the cut-free calculus for $\MDL$ introduced with the aid of the general method 
from~\cite{Lellmann:2013fk}.

\medskip
\noindent Through the paper we refer to the following {\Mimamsa} texts: the 
\emph{P\={u}rva {\Mimamsa} S\={u}tra} (henceforth 
PMS, ca. 3rd c. BCE), its commentary, the
\emph{\'{S}\={a}barabh\={a}\d{s}ya} (\'{S}Bh), the main subcommentary, Kum\={a}rila's \textit{Tantrav\={a}rttika} (TV).

\paragraph{Related work.}
Logic (mainly classical) has already been successfully used to investigate
other schools of Indian thought. 
In particular for  Navya Ny\={a}ya
formal analyses have contributed to a fruitful exchange of ideas
between disciplines~
\cite{Ganeri2008},
however, no deontic modalities were considered. 
A logical analysis of the deontic aspects of the \emph{Talmud}, another
sacred text, is given
in
\cite{Gabbay}. 
The deontic
logic used there is based on intuitionistic logic and
contains an external mechanism for resolving conflicts among obligations. 
Deontic logics similar but not equivalent to $\bMDL$ 
include Minimal Deontic Logic 
\cite{Goble:2013} and extensions of monotone modal logic
with some versions of the $\D$ axiom 
\cite{Indrzejczak:2005,Orlandelli:2014}. 
The latter papers 
also introduce cut-free sequent calculi, but do not mix alethic and deontic modalities.

\section{Extracting a deontic logic from {\Mimamsa} texts}
\label{sec:extraction}
The use of logic to simulate {\Mimamsa} ways of reasoning is motivated by
their rigorous theory of inference  
and attention for possible
violations  of it. For instance Kum\={a}rila, one of the chief {\Mimamsa} authors, emphasises the fact that a text
is not epistemically reliable if the whole chain of transmission is reliable, but not its beginning. The classical example
is that of ``\textit{a chain of truthful blind people transmitting information concerning colours}'' (TV on PMS 1.3.27).

At this point, the problem amounts to which logic should be adopted. 
The simplest logical system
for dealing with obligations
is 
\emph{Standard Deontic Logic} $\mathsf{SDL}$, that
 extends classical logic by
 a unary operator $\Ob$ read as
 ``It is obligatory that...'' satisfying the axioms of
modal logic $\K\D$ \cite{Blackburn:2001fk,Gab:2013}. Though
simple and well studied, $\mathsf{SDL}$ is not suited to deal with conflicting obligations,
which are often present in the Vedas and in {\Mimamsa} reasoning.
A well known example from the Vedas consists of the following norms concerning the
malefic \'Syena sacrifice, which is enjoined in case one desires to harm his enemy, since it kills them: 
\begin{quote}

A. ``\textit{One should not harm any living being}"

B. ``\textit{One should sacrifice bewitching with the \'{S}yena}" 
\end{quote}
Any reasonable formalisation of the statements A.\ and B.\ leads in $\mathsf{SDL}$ to a contradiction.
\label{nonSDL}
Given that the {\Mimamsa} authors embraced the principle of non-contradiction and invested all their efforts
in creating a \emph{consistent} deontic system, to
provide adequate formalisations of {\Mimamsa} reasoning a different logic is needed. To this aim we introduce
\emph{basic {\Mimamsa} deontic logic} ($\bMDL$)
by extracting its properties directly from {\Mimamsa} texts.

The language of $\bMDL$ extends that of classical logic with 
the binary modal operator $\Ob( \cdot / \cdot )$ from dyadic deontic logics and the unary modal operator $\NEC$ of $\Sf$.
While the latter is used to formalise the auxiliary conditions of general deontic principles,
the former allows us to impose conditions on obligations describing the situation in which the obligation holds. 
Hence a formula $\Ob( \varphi  / \psi )$ can be read as ``$\varphi$ it is obligatory given $\psi$''.

The use of the dyadic operator, which is a reasonably standard approach to avoid
the problem with conflicting obligations (see, e.g., \cite{Hilpinen2001} and \cite{Goble:2013}),
is also suggested in the metarule ``\textit{Each action is prescribed in relation to a responsible person 
who is identified because of her desire}'' (cf. PMS 6.1.1--3).

As described in Sec. \ref{sec:extracting-axioms} the
properties of the deontic operator $\Ob( \cdot / \cdot )$ of
$\bMDL$ (definition below) are directly extracted from the \emph{ny\={a}ya}s.

\begin{definition}\label{def:logic}
Basic {\Mimamsa} deontic logic $\mathsf{bMDL}$ extends (any Hilbert system for) $\Sf$ with the following axioms (taken as schemata):
\begin{enumerate}[label=(\arabic*)]
\item \label{ax:1}   $(\NEC (\varphi \rarr \psi ) \land  \Ob( \varphi / \theta ) ) \rarr  \Ob(\psi / \theta)$ 
\item \label{ax:2}$\NEC (  \psi  \rarr \neg \varphi ) \rarr \neg ( \Ob ( \varphi /  \theta ) \land \Ob ( \psi  /  \theta ))$
\item \label{ax:3} $(\NEC (( \psi  \rarr  \theta ) \land ( \theta  \rarr  \psi ) ) \land \Ob ( \varphi  /  \psi )) \rarr \Ob ( \varphi  /  \theta ) $
\end{enumerate}
\end{definition}
The choice to use classical logic as base system, in contrast to the use of intuitionistic logic in Gabbay et al.'s deontic
logic of the Talmud \cite{Gabbay}, is due to
various metarules by {\Mimamsa} authors implying the legitimacy of the reductio ad absurdum argument RAA;
these include the following (contained in Jayanta's book \emph{Ny\={a}yama\~{n}jar\={\i}}):
\label{noncontradiction}
``\textit{When there is a contradiction ($\varphi$ and not $\varphi$), at the denial of one (alternative), the other is known (to be true)}''.
Therefore, if we deny $\neg \varphi$ then $\varphi$ holds, which gives RAA.
\subsection{From {\Mimamsa} \emph{ny\={a}ya}s to Hilbert axioms}
\label{sec:extracting-axioms}
Axiom~\ref{ax:1} arises from three different principles, discussed in \cite{Atinerpaper}; among them the following abstraction 
of the \emph{ny\={a}ya}s in the {\it Tantrarahasya} IV.4.3.3 (see \cite{Freschi2012})
\begin{quote}
 If the accomplishment of X presupposes the
 accomplishment of Y, the obligation to perform X prescribes also Y.
\end{quote} This principle
leads to $ (\NEC (\varphi \rarr \psi ) \land  \Ob( \varphi / \theta )
)\rarr  \Ob(\psi / \theta)$, where we represent the accomplishment of
X and Y as $\varphi$ and $\psi$ respectively, and we stipulate that the conditions on the two prescriptions, represented by $\theta$, are the same.
Note that we use the operator $\NEC$, here as well as in the following axioms, to guarantee that the correlations between formulae are not accidental.

Axiom~\ref{ax:2} arises from the so-called {\em principle of the half-hen}, which is 
implemented in different {\Mimamsa} 
contexts (e.g., TV on PMS 1.3.3); an abstract representation of it is:
\begin{quote}
Given that purposes Y and Z exclude each other, if one should use item X for the purpose Y, then it cannot be the case that one should use it at the same time for the purpose Z.
\end{quote}
This principle stresses the incongruity of enjoining someone to act in
contradiction with himself on some object. The corresponding axiom is
\mbox{$\NEC (  \psi  \rarr \neg \varphi ) \rarr \neg ( \Ob ( \varphi /
  \theta ) \land \Ob ( \psi  /  \theta ))$} which guarantees that if $\varphi$ and $\psi$ exclude each other, then they cannot both be 
obligatory under the same conditions $\theta$.
Finally, Axiom~\ref{ax:3} arises from a discussion (in \'{S}Bh on PMS 6.1.25) on the eligibility to perform sacrifices (see \cite{Atinerpaper}),
which can be abstracted as follows:
\begin{quote}
If conditions X and Y are always equivalent, given the duty to perform Z under the condition X, the same duty applies under Y.
\end{quote}
We formalise this principle
as $(\NEC (( \psi  \rarr  \theta ) \land ( \theta  \rarr  \psi ) )
\land \Ob ( \varphi  /  \psi )) \rarr \Ob ( \varphi  /  \theta ) $,
where the conditions X and Y are represented by $\psi$ and $\theta$ respectively,
and $\varphi$ represents that the action Z is performed.

While the  properties of $\Ob( \cdot / \cdot )$ are taken from {\Mimamsa} texts, the same cannot be done for $\NEC$ because
{\Mimamsa} authors do not conceptualise necessity as separate from epistemic certainty. The established choices for a logic for the alethic necessity operator $\NEC$ are $\Sf$ and $\Sfi$. To keep the system as simple as possible, and not having found
any principle motivating the additional properties of
$\Sfi$, we have chosen $\Sf$.

\section{Proof Theory of $\bMDL$}
\label{sec:sequent-calculus}

Hilbert systems are 
convenient ways of defining logics, but are not very useful for 
proving theorems in and about the logics (e.g., decidability, consistency).

For this purpose  we introduce 
a cut-free sequent calculus $\GDL$ for
$\bMDL$ and use it to show that, for certain issues, $\bMDL$ simulates {\Mimamsa} ways of reasoning. 
As usual, a \emph{sequent} is a tuple $\Gamma \seq \Delta$ of multisets of formulae
interpreted as $\Land \Gamma
\to \Lor \Delta$.
To construct $\GDL$ we use the
translation from axioms to rules and the construction of 
a
cut-free calculus from these rules from~\cite{Lellmann:2013,Lellmann:2013fk}. Since the latter is not fully
automatic, we provide some details. 

First, by
\cite[Thm.~26]{Lellmann:2013fk}, we
automatically obtain from 
Def.~\ref{def:logic}\ref{ax:1}-\ref{ax:3}
the rules
  \[  \infer[\Mon ']{ \NEC \varphi, \Obl (\psi/\theta) \seq \Obl (\chi/\xi)}{ \varphi, \psi \seq \chi \qquad \seq \varphi, \psi \qquad \chi \seq \varphi \qquad \theta \seq \xi \qquad \xi \seq \theta }
\]
\[ \infer[\Cg]{ \NEC \varphi, \Obl (\psi/\theta) \seq \Obl (\chi/\xi)}{ \varphi, \theta \seq \xi  \quad \varphi, \xi \seq  \theta \quad \seq \varphi, \theta, \xi \quad \theta,\xi  \seq \varphi \quad \psi \seq \chi \quad \chi \seq \psi } 
\]
\[
\infer[\D_{2} ']{ \NEC \varphi , \Obl (\psi/\theta) , \Obl (\chi/\xi) \seq }{ \varphi,\psi,\chi \seq \quad \seq \varphi,\psi \quad  \seq \varphi, \chi  \quad  \theta \seq \xi \quad  \xi \seq \theta}
  \]
From these rules we 
construct 
a new set of rules saturated under cuts from which 
the rules above are
derivable. This step is not automatic and amounts to repeated \emph{cutting between
rules} 
\cite[Def.~7]{Lellmann:2013fk}: given any two rules
 we obtain a new rule whose conclusion is the result of a cut on a
 formula principal in the conclusions of both rules, and whose premisses
 contain
 all possible cuts between the premisses of the original rules on the
 variables occurring in this formula. We start from the set containing the rules above and those of $\mathsf{S}\4$ and
first cut the rules $\4$ (Fig.~\ref{fig:sequent-rules}) with $\Mon'$
and $\4$ with $\Cg$ on the boxed formula to 
obtain the rules
\[
  \infer[]{ \Gamma , \Obl (\psi/\theta)\seq \Obl (\chi/\xi) , \Delta }
  {   \Gamma^{\Box} ,\psi \seq \chi  & \theta \seq \xi & \xi \seq \theta  }
  \qquad
  \infer[]{\Gamma, \Obl (\psi / \theta) \seq \Obl(\chi / \xi),\Delta
  }
  {\Gamma^\Box, \theta \seq \xi & \Gamma^\Box, \xi \seq \theta & \psi
    \seq \chi & \chi \seq \psi
  }
\]
where 
$\Gamma^\Box$ is 
obtained
from $\Gamma$ by deleting every occurrence of a formula not of the
form $\Box \varphi$. Now cutting these two rules in either possible way yields the
rule $\Mon$ 
(Fig.~\ref{fig:sequent-rules}), and cutting this and  $\4$ 
with $\D_2'$ yields $\D_2$. We obtain $\D_1$ closing $\D_2$ under contraction, i.e., identifying $\varphi$ with $\theta$ and $\psi$ with $\chi$ and contracting conclusion and premiss.

The sequent calculus $\GDL$ consists of
the rules in Fig.~\ref{fig:sequent-rules} together with the standard
propositional $\G3$-rules (with principal formulae copied into the
premisses) \cite{Kleene:1952fk} and the standard left rule for the constant $\bot$.
We write $\deriv_\GDL \Gamma \seq \Delta$ if
 $\Gamma \seq \Delta$ is derivable using these
rules.
We denote extensions of $\GDL$ with structural rules from
Fig.~\ref{fig:structural-rules} by appending their names, collecting
$\Conl$ and $\Conr$ into $\Con$. E.g., $\GDL \Con\W$ is $\GDL$ extended
with Contraction and Weakening.
\begin{figure}[t]
\hrule\medskip
  \[  \infer[\4]{ \Gamma \seq \Box \varphi,\Delta}{ \Gamma^\Box \seq \varphi}
  \qquad
  \infer[\T]{\Gamma, \Box \varphi \seq \Delta}{\Gamma, \Box \varphi,\varphi \seq \Delta}
  \qquad
  \infer[\Mon]{\Gamma, \Obl (\varphi / \psi) \seq \Obl (\theta / \chi),\Delta}{ \Gamma^\Box, \varphi
    \seq \theta & \Gamma^\Box, \psi \seq \chi &  \Gamma^\Box, \chi \seq \psi}
  \]
  \[
  \infer[\D_1]{ \Gamma, \Obl( \varphi / \psi ) \seq \Delta }{ \Gamma^\Box, \varphi \seq \;}
  \qquad
  \infer[\D_2]{ \Gamma, \Obl (\varphi / \psi), \Obl (\theta / \chi) \seq \Delta}{ \Gamma^\Box,
    \varphi, \theta \seq \; & \Gamma^\Box, \psi \seq \chi & \Gamma^\Box, \chi \seq \psi}
  \]
\hrule
\caption{The modal rules rules of $\GDL$}
\label{fig:sequent-rules}
\end{figure}

By construction \cite{Lellmann:2013,Lellmann:2013fk} we
have: 
\begin{theorem}\label{thm:cut-elim}
  The rule $\Cut$ is admissible in $\GDL\Con\W$.
\end{theorem}
\begin{proof}
  Using the structural rules the system $\GDL\Con\W$ is equivalent to
  the system $\GDL'\Con\W$
  in which the principal formulae of the propositional rules and the
  rule $\T$ are
  not copied into the premisses. 
  By construction (and straightforward inspection in the
  non-principal cases) the rules of $\GDL'\Con\W$ satisfy the general sufficient
  criteria for
  cut elimination established 
  in
  \cite{Lellmann:2013,Lellmann:2013fk}. Cut-free derivations in
  $\GDL'\Con\W$ are converted into cut-free derivations in $\GDL\Con\W$ using the
  structural rules. \qed
\end{proof}
\begin{figure}[t]
  \hrule
  \centering
  \[
  \infer[\W]{\Gamma,\Sigma \seq \Delta,\Pi}{\Gamma \seq \Delta}
  \;\;
  \infer[\Conl]{\Gamma,\varphi \seq \Delta}{\Gamma, \varphi, \varphi \seq \Delta}
  \;\;
  \infer[\Conr]{\Gamma \seq \varphi,\Delta}{\Gamma \seq \varphi,\varphi, \Delta}
  \;\;
  \infer[\Cut]{\Gamma,\Sigma \seq \Delta, \Pi}{\Gamma \seq \varphi,\Delta &
    \Sigma, \varphi \seq \Pi}
  \]
  \hrule\medskip
  \caption{The structural rules}
  \label{fig:structural-rules}
\end{figure}

The methods in~\cite{Lellmann:2013,Lellmann:2013fk}
now automatically yield also an $\exptime$-complexity result. However,
we consider 
an explicit proof search procedure for $\GDL$ which will be used 
in Sec.~\ref{sec:semantics}.
First 
we establish some preliminary results.

\begin{lemma}\label{lem:adm-con}
  The Contraction and Weakening rules are admissible in $\GDL$.
\end{lemma}

\begin{proof}
  Admissibility of weakening is proved by induction on the depth of the derivation, while
  that of contraction follows from the general criteria in \cite[Thm.~16]{Lellmann:2013fk}
  resp. \cite[Thm.~2.5.5]{Lellmann:2013} since the rule set $\GDL$ is contraction
  closed and already contains the modified versions of $\T$
  and the propositional rules.\qed
\end{proof}

Thus suffices to consider \emph{set-based sequents}, i.e.,
tuples of sets of formulae instead of multisets. The rules of $\GDL$
are adapted to the set-based setting in the standard way.
Since boxed
formulae are always copied into the premisses of a rule, the
proof search procedure needs to include \emph{loop
  checking} to avoid infinite branches in the search tree. We do this
using 
histories, i.e., lists of (set-based)
sequents, where the last element is interpreted as the current sequent:

\begin{definition}[Histories]
  A \emph{history} $\mathcal{H}$ is a finite list $[\Gamma_1 \seq
  \Delta_1; \dots ; \Gamma_n \seq \Delta_n]$ of set-based sequents,
  where we write $\lastl{\mathcal{H}}$ (resp.\ $\lastr{\mathcal{H}}$)
  for $\Gamma_n$ (resp.\ $\Delta_n$) and $\last{\mathcal{H}}$ for
  $\lastl{\mathcal{H}} \seq \lastr{\mathcal{H}}$. Given another
  history $\mathcal{H}' = [\Sigma_1 \seq \Pi_1; \dots ; \Sigma_m \seq
  \Pi_m]$ with $n \leq m$ we write $\mathcal{H} \cless \mathcal{H}'$ if for all $i\leq
  n$ we have $\Gamma_i = \Sigma_i$ and $\Delta_i = \Pi_i$. Finally, we
  write $\mathcal{H}\plusplus \mathcal{H}'$ for the concatenation of
  the two histories. 
\end{definition}

The proof search procedure for $\GDL$ is given in Algorithm~\ref{alg:proof-search}, where following \cite{Gore:1999} we call the propositional rules together with the rule $\T$ the
\emph{static} rules, $\Mon,\4,\D_1,\D_2$ are
called \emph{transitional} rules. 
The algorithm
saturates the current sequent under backwards applications of the one-premiss
static rules, and then checks whether the result is an initial sequent
or could have been derived by a two-premiss static rule or a dynamic rule. The histories are used
to prevent the procedure from exploring a sequent twice (modulo weakening).

\begin{algorithm}[t]\label{alg:proof-search}
  \caption{The proof search procedure for $\GDL$}
  \KwIn{A history $\mathcal{H}$}
  \KwOut{Is $\last{\mathcal{H}}$ derivable in $\GDL$ given the history
    $\mathcal{H}$?}
  \BlankLine
  Saturate $\last{\mathcal{H}}$ under the one-premiss static rules\;
  \eIf{$\last{\mathcal{H}}$ is an initial sequent}{accept the history}
  {\For{every possible application of a two-premiss static
      rule to $\last{\mathcal{H}}$}
    {\For{every premiss $\Sigma \seq \Pi$ of this application}
      {recursively call the proof search procedure with input
        $\mathcal{H}\plusplus[\Sigma \seq \Pi]$\;
      }
    accept the application if each of these calls accepts 
  }
  \For{every possible application of a transitional rule
    to 
    $\last{\mathcal{H}}$}
  {\For{every premiss $\Sigma \seq \Pi$ of this application}
    {\eIf{there is an 
        $\mathcal{H}'\cless \mathcal{H}$
        with $\Sigma \subseteq
        \lastl{\mathcal{H}'}$ and $\Pi \subseteq
        \lastr{\mathcal{H}'}$
      }
      {reject the premiss
      }
      { 
        call the proof search procedure with input
      $\mathcal{H} \plusplus [\Sigma \seq \Pi]$\;
      accept the premiss if this call accepts
      }
    }
  accept the rule application 
  if each of the premisses is accepted 
  }
  accept the history if at least one of the possible applications is
  accepted
}
\end{algorithm}

\begin{lemma}[Termination]\label{lem:termination}
  The proof search procedure terminates.
\end{lemma}

\begin{proof}
  Given a history $\mathcal{H}$, the number $N$ of different set-based
  sequents which can be constructed from subformulae of the sequent
  $\last{\mathcal{H}}$ is exponential in the size of
  $\last{\mathcal{H}}$. Hence after at most $N$-many recursive calls
  of the 
  procedure the subroutine rejects every rule 
  application.  
  Furthermore, for every sequent there are only
  finitely many possible (backwards) applications of a rule from
  $\GDL$, so 
  the subroutine is executed only 
  a finite number of times.\qed
\end{proof}

\begin{proposition}\label{prop:correctness}
$\deriv_\GDL \Gamma \seq \Delta$ 
  iff the procedure accepts 
  $[\Gamma \seq \Delta]$.
\end{proposition}

\begin{proof}
  If the procedure accepts the input,
  then we
  construct a derivation of $\Gamma \seq \Delta$ in $\GDL$
  by following the accepting choices of backwards applications of the
  rules, and labelling the nodes in the derivation with the sequents
  $\last{\mathcal{H}}$ for the histories $\mathcal{H}$ given as input
  to the recursive calls of the algorithm.

  Conversely, 
  if the set-based sequent $\Gamma
  \seq \Delta$ is derivable in $\GDL$, then by admissibility of
  Weakening there is a \emph{minimal} derivation of it, i.e., a
  derivation in which no branch contains two set-based sequents
  $\Sigma \seq \Pi$ and $\Omega \seq \Theta$ such that $\Sigma \seq
  \Pi$ occurs 
  on the path between $\Omega \seq \Theta$ and the root,
  and such that $\Omega\subseteq\Sigma$ and $\Theta \subseteq\Pi$. By
  induction on the depth of such a minimal derivation it can then be
  seen that the procedure accepts the input $[\Gamma \seq \Delta]$.\qed
\end{proof}

\subsection{Inner and Outer Consistency}
\label{sec:in-out-cons}

Having extracted a 
cut-free calculus from the axioms 
using the method in \cite{Lellmann:2013,Lellmann:2013fk},
soundness and completeness w.r.t. $\bMDL$ follow by
construction (Thm.~\ref{thm:strong-compl}). By the subformula property we then 
obtain
the \emph{inner
  consistency} of the logic $\bMDL$, i.e., the fact that $\bot$ is not
a theorem 
of the logic. This is one of the most basic requirements that our
logic 
should satisfy. 
But since $\bMDL$ was introduced with the purpose of simulating {\Mimamsa} reasoning,
it should also be consistent with
respect to the examples considered by the {\Mimamsa} authors such as
the \'Syena sacrifice, i.e.,
it should not enable us to derive a contradiction from the
formalisations of these
examples. We capture this in the notion of \emph{outer consistency}
or consistency 
in presence
of global assumptions. 
To make this precise
we consider the
consequence relation associated with the logic $\bMDL$ and the
corresponding relation associated with the
calculus $\GDL$.
Henceforth we denote by $\mathcal{A}$ any set of formulae of $\bMDL$.

\begin{definition}
The usual notion of derivability of a formula $\varphi$ from a set
$\mathcal{A}$ of assumptions in $\bMDL$ is denoted by $\mathcal{A} \deriv_\bMDL\varphi$.
  Similarly, for a set $\mathcal{S}$ of
  sequents, a sequent $\Gamma \seq \Delta$ is \emph{derivable from
    $\mathcal{S}$ in $\GDL\Cut$} if there is a derivation of $\Gamma
  \seq \Delta$ in $\GDL$ with leaves labelled with initial sequents,
  zero-premiss rules 
  or sequents from 
  $\mathcal{S}$. We then write $\mathcal{A} \deriv_\bMDL \varphi$
  resp.\ $\mathcal{S} \deriv_{\GDL\Cut} \Gamma \seq \Delta$.
\end{definition}

\begin{theorem}[Soundness and Completeness]\label{thm:strong-compl}
  For all sets $\mathcal{S}$ of sequents and sequents $\Gamma \seq
  \Delta$ 
  we have: 
  \[\mathcal{S} \deriv_{\GDL\Cut} \Gamma \seq \Delta \; \; \text{ iff }\;\; \{
    \Land \Sigma \to \Lor \Pi \mid \Sigma \seq \Pi \in \mathcal{S} \}
    \deriv_{\bMDL} \Land \Gamma \to \Lor \Delta\;.\]
\end{theorem}

\begin{proof}
  The corresponding standard results for the propositional calculi
  transfer readily to the system $\bMDL$ and the
  Gentzen system $\G3$ with the zero-premiss rules  
  $\;\vcenter{
    \infer{\;\seq \theta}{\phantom{\;\seq\theta}}}\;$ for each modal
  axiom schema $\theta$ of $\bMDL$. The result then follows from interderivability
  of these rules
  with the modal rules from $\GDL$
  \cite{Lellmann:2013,Lellmann:2013fk}. As an example, the derivation of the zero-premiss rule for Axiom \ref{ax:2}, where $\alpha$ denotes \mbox{$\NEC (  \psi  \rarr \neg \varphi ) \rarr \neg ( \Ob ( \varphi /  \theta ) \land \Ob ( \psi  /  \theta ))$,} is as follows
  \[
\infer=[prop.]{\; \seq \strut \alpha}{\infer[\D_2]{\Ob ( \psi  /  \theta ), \Ob ( \varphi /  \theta ) , \Ob ( \varphi /  \theta ) \land \Ob ( \psi  /  \theta ), \NEC (  \psi  \rarr \neg \varphi ) \seq \alpha, \neg ( \Ob ( \varphi /  \theta ) \land \Ob ( \psi  /  \theta ))}{ \infer*{\NEC (  \psi  \rarr \neg \varphi ) , \psi , \varphi \seq \;}{{\mathcal D}_{1}}
 &
\infer[ax.]{\NEC (  \psi  \rarr \neg \varphi ) , \theta \seq \theta}{} &
\infer[ax.]{\NEC (  \psi  \rarr \neg \varphi ) , \theta \seq \theta}{} }}  
  \]
where the double line denotes multiple applications of the propositional rules and the derivation ${\mathcal D}_{1}$ is
\[\infer[\T]{  \NEC (  \psi  \rarr \neg \varphi ) , \psi , \varphi \seq \;}{\infer[\rarr \seq]{\psi  \rarr \neg \varphi  ,\NEC (  \psi  \rarr \neg \varphi ) , \psi , \varphi \seq \;}{ \infer[ax.]{\psi  \rarr \neg \varphi  ,\NEC (  \psi  \rarr \neg \varphi ) , \psi , \varphi \seq \psi}{} & \infer[\neg \seq]{\neg \varphi   , \psi  \rarr \neg \varphi  ,\NEC (  \psi  \rarr \neg \varphi ) , \psi , \varphi \seq \;}{\infer[ax.]{\neg \varphi   , \psi  \rarr \neg \varphi  ,\NEC (  \psi  \rarr \neg \varphi ) , \psi , \varphi \seq \varphi}{} }}}\]\qed
\end{proof}

\begin{corollary}
The logic $\bMDL$ is consistent, i.e., $\bot \not\in \bMDL$.\qed
\end{corollary}

\begin{proof}
Follows by Thm. \ref{thm:strong-compl}.1 and the fact that the rules of $\GDL$ satisfy the subformula property.\qed
\end{proof}

\begin{definition}
  $\bMDL$ enjoys \emph{outer consistency} with respect to 
  $\mathcal{A}$ if
  $\mathcal{A} \not \deriv_{\bMDL} \bot$
\end{definition}

By Thm.~\ref{thm:strong-compl} this 
condition 
is equivalent to $\{ \;\seq \varphi \mid \varphi \in
\mathcal{A}\} \not \deriv_{\GDL\Cut} \;\seq \bot$.
We now show that $\bMDL$ allows us to consistently formalise the seemingly conflicting
statements 
of the \'Syena sacrifice. 
The proof uses the 
proof search procedure given in Algorithm~\ref{alg:proof-search}
and 
the following version of the deduction theorem (see Section \ref{sec:logic-at-work} for a semantic proof).

\begin{theorem}\label{thm:deduction}
For every sequent $\Gamma \seq \Delta$ and set $\mathcal{A}$ of formulae the following are equivalent (writing $\Box {\mathcal A}$ for $\{ \Box\varphi \mid \varphi \in \mathcal{A} \}$ taken as a multiset):
  \begin{enumerate}
  \item $\{ \;\seq \varphi \mid \varphi \in \mathcal{A}\} \deriv_{\GDL\Cut} \Gamma \seq \Delta$\label{item:deriv1}
  \item $\{ \;\seq \Box \varphi \mid \varphi \in \mathcal{A}\} \deriv_{\GDL\Cut} \Gamma \seq \Delta$\label{item:deriv2}
  \item $\deriv_\GDL \Box \mathcal{A},\Gamma \seq \Delta$.\label{item:deriv3}
  \end{enumerate}
\end{theorem}

\begin{proof}
  $\ref{item:deriv1} \to \ref{item:deriv2}$: 
Easily follows by using the rules $\T$ and $\Cut$.

  $\ref{item:deriv2} \to \ref{item:deriv3}$: Since every rule in
  $\GDL$ 
  copies all boxed formulae in the
  antecedent 
  from
  conclusion to 
  premisses, the result of adding
  the formulae $\{ \Box \varphi \mid \varphi \in \mathcal{A} \}$ to the
  antecedents 
  of every 
  sequent occurring in the derivation of $\Gamma
  \seq \Delta$ from $\{ \Box \varphi \mid \varphi \in \mathcal{A}\}$ is still a
  derivation. As this turns every assumption 
$\;\seq \Box
  \varphi$ 
  into the
  derivable sequent $\Box \mathcal{A} \seq \Box \varphi$, the result is a
  derivation without assumptions. Statement \ref{item:deriv3} now
  follows using Cut Elimination (Thm.~\ref{thm:cut-elim}).

  $\ref{item:deriv3} \to \ref{item:deriv1}$: 
Easily follows by using the rules $\4$ and $\Cut$.\qed
\end{proof}

Thus in order to check whether $\bMDL$ enjoys outer consistency w.r.t.
a set $\mathcal{A}$ of formulae it is sufficient to check that the
sequent $\Box \mathcal{A} \seq \bot$ is not derivable in $\GDL$.
Before we 
formalise the \'Syena sacrifice, 
let us remark that while
the 
operator $\Ob(\cdot / \cdot )$ 
only
captures \emph{conditional obligations}, we would also like to reason
about \emph{unconditional obligations}, i.e., obligations which always
have to be fulfilled. We formalise such 
obligations in the standard way 
by 
$\Ob(\cdot / \top)$. 
A formula $\Ob(\varphi / \top)$ then can be read
as ``it is obligatory that $\varphi$ provided \emph{anything} is the case'', and
thus 
models an unconditional obligation. 
A formalisation of the problematic
situation in the \'Syena example (sentences A. and B. in Sec.~\ref{sec:extraction}) 
then is:
\begin{enumerate}
\item $\Ob (\neg \harm / \top)$ for ``One should not perform violence on any living being"\label{not_harm}
\item $\Ob(\syena/\desharm)$ for ``If you desire to harm your enemy you should perform the \'Syena''
   \label{syena}
\item $\harmenemy \to
\harm$ for ``harming the enemy entails harming a living
  being''
\item $\syena \to
\harmenemy$ for ``performing the \'Syena  entails harming the enemy''.\label{item:syenatoharm}
\end{enumerate}
with the variables $\harm$ for ``performing violence on any living
being'', $\syena$ for ``performing the \'Syena sacrifice'',
$\harmenemy$ for ``harming your enemy'', and $\desharm$ for ``desiring to harm your enemy''.

\begin{theorem}
  \label{thm:syena}
  $\bMDL$ enjoys outer consistency w.r.t. 
  the
  \'Syena sacrifice, i.e.:
  \[
  \bigl\{\; \harmenemy \to \harm,\; \syena\to\harmenemy,\; \Obl(\neg
  \harm / \top),\; \Obl(\syena / \desharm) \;\bigr\} \not\deriv_\bMDL \bot\;.
  \]
\end{theorem}

\begin{proof}
  By 
  Thm.~\ref{thm:strong-compl} and Thm.~\ref{thm:deduction}
  it is sufficient to show 
  that the sequent
  \[
  \label{form:syena-sequent}
  \Box (\harmenemy \to \harm), \Box (\syena\to\harmenemy), \Box\Obl(\neg
  \harm / \top), \Box\Obl(\syena / \desharm) \seq \bot
  \]
  is not derivable in $\GDL$. This is 
  done in the standard way by
  (a bit tediously) performing an exhaustive proof search following the 
  procedure
  in Algorithm~\ref{alg:proof-search}.\qed
\end{proof}

\section{Semantics of $\bMDL$}
\label{sec:semantics}

The semantics for $\bMDL$ is build on the standard semantics for modal logic $\Sf$, i.e.,
Kripke-frames with transitive and reflexive accessibility
relation \cite{Blackburn:2001fk}. The
additional modality $\Obl$ is captured using 
\emph{neighbourhood semantics}~\cite{Chellas:1980fk}, which 
we
modify 
to take into account only accessible worlds.
Intuitively, the neighbourhood map singles out a set of deontically
acceptable sets of accessible worlds for certain possible situations,
i.e., sets of accessible worlds.
As usual, if $R
\subseteq W\times W$ is a
relation and $w \in W$, we write $R[w]$ for $\{ v \in W \mid
wRv\}$. Also, for a set $X$ we write $X^c$ for the complement of $X$
(relative to an 
implicitly given set).

\begin{definition}\label{def:m-frame}
  A \emph{{\Mimamsa}-frame} (or briefly:
  \emph{m-frame}) is a triple
  $(W,R,\eta)$ consisting of a non-empty set $W$ of \emph{worlds} or
  \emph{states}, an \emph{accessibility relation} $R \subseteq W
  \times W$ and a map $\eta: W \to \Pow( \Pow (W) \times \Pow(W))$
  such that:
  \begin{enumerate}
  \item $R$ is transitive and reflexive;\label{item:one}
  \item if $(X,Y) \in \eta(w)$, then $X \subseteq R[w]$ and $Y
    \subseteq R[w]$;\label{item:two}
  \item if $(X,Z) \in \eta(w)$ and $X \subseteq Y\subseteq R[w]$, then also
    $(Y,Z)\in \eta(w)$;\label{item:three}
  \item $(\emptyset,X) \notin \eta(w)$;\label{item:four}
  \item if $(X,Y) \in \eta(w)$, then $(X^c\cap R[w],Y)\notin \eta(w)$.\label{item:five}
  \end{enumerate}
  A \emph{{\Mimamsa}-model} (or \emph{m-model}) is a m-frame 
  with a \emph{valuation} $\sigma: W \to \Pow(\Var)$.
\end{definition}

Intuitively, Condition~\ref{item:one} in Def.~\ref{def:m-frame} 
corresponds to 
axioms $(\4)$
and $(\T)$ of $\Sf$, Condition~\ref{item:two} ensures that only accessible
worlds influence the truth of a formula $\Obl(\varphi / \psi)$ and comes
from 
the rules $(\Mon)$ and $(\Cg)$,
Condition~\ref{item:three} corresponds to the rule
$(\Mon)$, while 
Conditions~\ref{item:four} resp.~\ref{item:five} correspond to 
$(\D_1)$ resp.\ 
$(\D_2)$. 

\begin{definition}[Satisfaction, truth set]
  Let $\mathfrak{M} = (W,R,\eta),\sigma$ be a m-model. The \emph{truth
  set} $\trs{\varphi}_{\mathfrak{M}}$ of a formula $\varphi$ in $\mathfrak{M}$ is defined recursively by
\begin{enumerate}
\item $\trs{p}_\mathfrak{M} \defs \{ w \in W \mid p \in \eta(w)\}$
\item $\trs{\Box \varphi}_\mathfrak{M} \defs \{ w \in M \mid R[w] \subseteq \trs{\varphi}_\mathfrak{M}\}$
\item $\trs{\Obl(\varphi / \psi)}_\mathfrak{M} \defs \{ w \in W \mid
  (\trs{\varphi}_\mathfrak{M} \cap R[w], \trs{\psi}_\mathfrak{M} \cap
  R[w]) \in \eta(w) \}$\label{item:truth-cond-obl}
\end{enumerate}
and 
the standard clauses for the boolean connectives. We
omit the subscript $\mathfrak{M}$ if the m-model is clear from
the context, and we 
write $\mathfrak{M},w \forces \varphi$ for $w
\in \trs{\varphi}_\mathfrak{M}$.
A formula $\varphi$ is
\emph{valid in} a m-model $\mathfrak{M}$ if for all worlds $w$ of
$\mathfrak{M}$ we have $\mathfrak{M},w \forces \varphi$.
\end{definition}

Note that in clause \ref{item:truth-cond-obl}
we
slightly deviate from the standard treatment in that we restrict the
attention to worlds accessible from the current world.

\begin{lemma}\label{lem:soundness-preservation}
  For all rules of $\GDL$ we have: if the interpretations of its
  premisses are valid in all m-models, then so is the interpretation
  of its conclusion.
\end{lemma}

\begin{proof}
  We show that if the negation of the interpretation of the conclusion 
  is
  satisfiable in a m-model, then so is the negation of the
  interpretation of (at least) one
  of the premisses.
  For $\4,\T$ and the propositional rules 
  this is standard. 

  For the modal rules we only show the case of $\D_2$, the other cases being similar.
  Assume that for the m-model $\mathfrak{M} =
  (W,R,\eta),\sigma$ the negation of the conclusion is satisfied in $w
  \in W$, i.e., we have
  $
  \mathfrak{M},\sigma \forces \Land  \Gamma \land \Obl(\varphi / \psi)
  \land \Obl (\theta / \chi)\;.
  $
  Then we have $(\trs{\varphi} \cap R[w],\trs{\psi} \cap R[w]) \in \eta(w)$ and
  $(\trs{\theta} \cap R[w], \trs{\chi} \cap R[w]) \in\eta(w)$. By
  Cond.~\ref{item:five} in Def.~\ref{def:m-frame} 
  we know
  that $(\trs{\varphi}^c  \cap R[w],\trs{\psi} \cap R[w]) \not\in \eta(w)$, hence
  $\trs{\theta} \cap R[w] \neq \trs{\varphi}^c  \cap R[w]$ or $\trs{\psi} \cap R[w]
  \neq \trs{\chi} \cap R[w]$. If the latter does not hold, using this
  and 
  Cond.~\ref{item:three} we have $\trs{\varphi}^c\cap
  R[w]\subsetneq\trs{\theta}\cap R[w] $ and hence 
  we find a world $v \in
  \trs{\varphi} \cap \trs{\theta} \cap R[w]$. Then with transitivity we obtain
  $
  \mathfrak{M},\sigma, v \forces \Land  \Gamma^\Box \land \varphi \land \theta
  $,
  and thus the negation of the first premiss of the rule is
  satisfiable. Otherwise
  we have $\trs{\psi}\cap \trs{\chi}^c \cap R[w] \neq \emptyset$ or $\trs{\chi}\cap \trs{\psi}^c \cap R[w] \neq \emptyset$ 
  and again using transitivity we satisfy the
  negation of the second or the third premiss of the rule. 
  \qed
\end{proof}

\begin{corollary}[Soundness of $\GDL$]
  For every sequent $\Gamma \seq
  \Delta$ we have: if $\deriv_\GDL \Gamma \seq \Delta$, then $\Land \Gamma \to
  \Lor \Delta$ is valid in all m-models.
\end{corollary}

\begin{proof}
  By induction on the depth of the derivation, using Lem.~\ref{lem:soundness-preservation}.\qed
\end{proof}

For completeness we show how to construct a countermodel for a given sequent from a failed proof search for 
it.
For this, fix $\Gamma \seq \Delta$ to be a sequent 
not derivable in $\GDL$. We 
build 
a m-model
$\mathfrak{M}_{\Gamma \seq \Delta} = (W,R,\eta),\sigma$ from a
rejecting run of Alg.~\ref{alg:proof-search} 
on input
$[\Gamma \seq \Delta]$, such
that $\Land \Gamma \land \Land \neg \Delta$ is satisfied
in a world 
of
$\mathfrak{M}_{\Gamma \seq \Delta}$. For this, 
take the set $W$ of worlds
to be the set of all histories occurring in the run of the
procedure. 
To define the accessibility relation 
we first
construct the intermediate relation $R'$ by setting
$\mathcal{H} R' \mathcal{H}'$ iff (at least) one of the following holds:
\begin{enumerate}
\item $\mathcal{H} \cless
                            \mathcal{H}'$; or
\item $\mathcal{H}' \cless \mathcal{H}$ and there is a transitional rule application with conclusion
$\last{\mathcal{H}}$ and a premiss $\Sigma \seq \Pi$ of this rule
application such that $\Sigma
\subseteq \lastl{\mathcal{H}'}$ and $\Pi \subseteq
\lastr{\mathcal{H}'}$.\label{item:loops}
\end{enumerate}
Intuitively, in \ref{item:loops}.\ 
we add the loops which have been detected by the 
procedure.
The relation $R$ then is defined as the reflexive and
transitive closure of $R'$.
To define the function $\eta$ we first introduce a syntactic version of the truth set
notation: 
\[
\sts{\varphi}_W \defs \left\{ \mathcal{H} \in W \mid \varphi \in \lastl{\mathcal{H}} \right\}
\]
Now we define $\eta : W \to \Pow(\Pow (W) \times
\Pow(W))$
by setting for every history $\mathcal{H}$ in $W$:
\[
\eta(\mathcal{H}) \defs \left\{ (X,Y) \in \Pow(R[\mathcal{H}])^2 \mid \begin{array}{ll}\text{ for
    some formula }\Obl(\varphi/\psi) \in \lastl{\mathcal{H}}:\\ \sts{\varphi}_W\cap R[\mathcal{H}] \subseteq X\text{ and
                                                           }\sts{\psi}_W\cap
                                                           R[\mathcal{H}]
                                                           = Y
                                                         \end{array}
                                                       \right\}\;.
\]
Finally, we define the valuation $\sigma$ 
by
setting for every 
variable $p \in \Var$:
\[
\sigma(p) \defs \sts{p}_W\;.
\]
Let us write $\mathfrak{M}_{\Gamma \seq \Delta}$ for the resulting
structure $(W,R,\eta)$. Then we have:

\begin{lemma}\label{lem:model}
  The structure $\mathfrak{M}_{\Gamma\seq\Delta},\sigma$ is a m-model.
\end{lemma}

\begin{proof}
  By construction 
  $\sigma$ is a valuation, 
  $R$ is a transitive and reflexive relation on $W$, and 
  Conditions~\ref{item:two} and~\ref{item:three} of
  Def.~\ref{def:m-frame} hold for $\eta$.
  To see that Condition~\ref{item:five} holds, we need to show that if 
  $(X,Y) \in \eta(\mathcal{H})$ then $(X^c \cap R[\mathcal{H}], Y)
  \not\in \eta(\mathcal{H})$.
  For this 
  we show that whenever $(X,Y) \in \eta(\mathcal{H})$ and
  $(Z,W) \in \eta(\mathcal{H})$, then $Z \not=
  X^c\cap R[\mathcal{H}]$
  or $Y \neq W$. So assume we have such $(X,Y)$ and $(Z,W)$ in
  $\eta(\mathcal{H})$. 
  By construction of $\eta$ there must be
  formulae
  $\Obl(\varphi/\psi)$ and $\Obl(\theta/\chi)$ in $\lastl{\mathcal{H}}$ such that
  \begin{itemize}
  \item $\sts{\varphi}_W \cap R[\mathcal{H}] \subseteq X$ and $\sts{\psi}_W
    \cap R[\mathcal{H}] = Y$; and 
  \item $\sts{\theta}_W \cap R[\mathcal{H}] \subseteq Z$ and $\sts{\chi}_W
    \cap R[\mathcal{H}] = W$.
  \end{itemize}
  Since both $\Obl(\varphi/\psi)$ and $\Obl(\theta/\chi)$ are in $\lastl{\mathcal{H}}$,
  the transitional rule $D_2$ can be applied to
  $\last{\mathcal{H}}$. Thus 
  the proof search procedure either
  used the
  premisses
  \[
  \lastl{\mathcal{H}}^\Box,\varphi,\theta \seq \;
  \qquad
  \lastl{\mathcal{H}}^\Box,\psi \seq \chi
  \qquad
  \lastl{\mathcal{H}}^\Box,\chi \seq \psi
  \]
  of this rule application to create new histories by appending them to
  $\mathcal{H}$, or it found a history $\mathcal{H}'\cless\mathcal{H}$
  whose last sequent subsumes one of the premisses. 
  In either case for at
  least one premiss $\Sigma \seq \Pi$ there is a history
  $\mathcal{H}'$ s.t. $\Sigma \subseteq \lastl{\mathcal{H}'}$ and
  $\Pi \subseteq \lastr{\mathcal{H}'}$ and for which proof search fails. Moreover, for this $\mathcal{H}'$ by construction of $R$ we know that
  $\mathcal{H}R\mathcal{H}'$. Assume that $\Sigma \seq \Pi$ is
  the first premiss. 
  Then
  $\varphi,\theta \in \lastl{\mathcal{H}'}$, and hence 
  $\mathcal{H}' \in \sts{\varphi}_W \cap \sts{\theta}_W \cap R[\mathcal{H}]$ and
  the latter is non-empty. Then in particular 
  $X^c \cap R[\mathcal{H}] \subseteq (\sts{\varphi}_W\cap R[\mathcal{H}])^c \cap
  R[\mathcal{H}] = (\sts{\varphi}_W)^c \cap R[\mathcal{H}]$ is not equal to
  $\sts{\theta}_W \cap R[\mathcal{H}] = Z$. Similarly, if $\Sigma \seq
  \Pi$ is 
  one of the remaining 
  premisses 
  we obtain $Y \neq W$. Thus whenever $(X,Y) \in \eta(\mathcal{H})$
  and $(Z,W) \in \eta(\mathcal{H})$, then $Z \neq X^c \cap
  R[\mathcal{H}]$ or $Y \neq W$. 
  The reasoning for Cond.~\ref{item:four} is similar.\qed
\end{proof}

\begin{lemma}[Truth Lemma]\label{lem:truth-lemma}
  For every history $\mathcal{H} \in W$:
(i) If $\varphi \in
  \lastl{\mathcal{H}}$, then $\mathfrak{M}_{\Gamma \seq
    \Delta},\sigma,\mathcal{H} \forces \varphi$ and (ii)
if $\psi \in \lastr{\mathcal{H}}$, then $\mathfrak{M}_{\Gamma \seq
    \Delta},\sigma,\mathcal{H} \forces\neg \psi$.
\end{lemma}

\begin{proof}
  We prove both statements simultaneously by induction on the
  complexity of 
  $\varphi$ resp.\ $\psi$. The base case and the cases where the
  main 
  connective of $\varphi$ resp.\ $\psi$ is a propositional 
  or $\Box$ are 
  standard (note that Alg.~\ref{alg:proof-search} 
  saturates every sequent under the static rules, i.e., the
  propositional rules and 
  $\T$, and that every
  transitional rule copies all the boxed formulae in the antecedent 
  into the premisses).
  If 
  $\varphi = \Obl(\theta /\chi)$, then by construction
  of $\eta$ we have 
  $(\sts{\theta}_W\cap R[\mathcal{H}],\sts{\chi}_W\cap
  R[\mathcal{H}]) \in \eta(\mathcal{H})$, and thus 
  $\mathfrak{M}_{\Gamma \seq \Delta},\sigma,\mathcal{H} \forces \Obl(\theta
  / \chi)$. Now suppose that $\psi =  \Obl(\xi / \gamma)$. To see
  that $\psi$ does not hold in $\mathcal{H}$ we show
  that for no $\Obl(\delta / \beta) \in \lastl{\mathcal{H}}$ we have
  $
  \sts{\delta}_W \cap R[\mathcal{H}] \subseteq \sts{\xi}_W \cap
  R[\mathcal{H}]
  $
   and 
  $
  \sts{\beta}_W \cap R[\mathcal{H}] = \sts{\gamma}_W \cap R[\mathcal{H}]\;.
  $
  The result then follows by construction of $\eta$ and the definition
  of 
  truth set. 
  If
  $\lastl{\mathcal{H}}$ does not contain any formula of the form
  $\Obl(\delta / \beta)$, then $\eta(\mathcal{H})$ is empty and we are done. Otherwise, there is such a 
  $\Obl(\delta / \beta)$ and the rule $\Mon$ can be applied backwards to
  $\last{\mathcal{H}}$. But then from the failed proof search
  for at least one of the 
  premisses
  \[
  \lastl{\mathcal{H}}^\Box,\delta \seq \xi
  \qquad
  \lastl{\mathcal{H}}^\Box,\gamma \seq \beta
  \qquad
  \lastl{\mathcal{H}}^\Box,
  \beta \seq \gamma
  \]
  we obtain a history $\mathcal{H}'$ with $\mathcal{H}R\mathcal{H}'$ whose last sequent subsumes this
  premiss. But then as above
  either $\sts{\delta}_W \cap R[\mathcal{H}] \not \subseteq
  \sts{\xi}_W \cap R[\mathcal{H}]$, if it is obtained 
  from the
  first premiss, or $\sts{\beta}_W \cap R[\mathcal{H}] \neq \sts{\gamma}_W \cap
  R[\mathcal{H}]$ otherwise. 
  \qed
\end{proof}

\begin{theorem}[Completeness]
\label{thm:compl-by-countermodel-constr}
  For every sequent $\Gamma \seq \Delta$ we have: if $\Land \Gamma
  \to \Lor \Delta$ is valid in every m-model, then $\deriv_{\GDL}
  \Gamma \seq \Delta$.
\end{theorem}

\begin{proof}
  If 
  $\not\deriv_{\GDL}\Gamma \seq \Delta$, 
  then by
  Lem.~\ref{lem:termination} and Prop.~\ref{prop:correctness} 
  the procedure in Alg.~\ref{alg:proof-search}
  terminates and rejects the input $[\Gamma \seq \Delta]$. Thus by
  Lem.~\ref{lem:model}
  and~\ref{lem:truth-lemma} we have $\mathfrak{M}_{\Gamma \seq
    \Delta},[\Gamma \seq \Delta] \forces \Land \Gamma \land \neg
  \Lor\Delta$ and hence $\Land\Gamma\to\Lor\Delta$ is not m-valid.\qed
\end{proof}

\noindent
Since only finitely many histories occur in a run of the proof search
procedure, the constructed model is finite and by standard methods we immediately obtain:

\begin{corollary}
  The logic $\MDL$ has the finite model property and is decidable.\qed
\end{corollary}

\section{Applications to Indology}
\label{sec:logic-at-work}
We show now that despite being reasonably simple, $\MDL$ is strong enough to 
derive consequences about topics discussed by {\Mimamsa} authors (Example \ref{ex:topics}) 
and to provide useful insights on the reason why the seemingly conflicting statements in the \'Syena example are not contradictory.

\begin{example}
\label{ex:topics}
Consider the following excerpt:
``\textit{Since the Veda is for the purpose of an action, whatever in it does not aim at an action is meaningless and therefore must be said not to belong to the permanent Veda}'' (PMS 1.2.1).
In other words: each Vedic prescription should promote an action. Given that no actual action can have a logical contradiction as an effect, a logical contradiction cannot be enjoined by an obligation. This can be translated into the formula $\neg \Ob ( \bot /  \theta )$,
one of the forms of axiom $D$,
which 
is derivable in $\G_{\MDL}$ as follows:
\[
\infer[\seq  \neg]{\; \seq \neg \Ob ( \bot /  \theta )}
{\infer[\D_1]{\Ob ( \bot /  \theta ) \seq \neg \Ob ( \bot /  \theta ) }{\infer[\bot \seq]{ \bot \seq\; }{ }}
}
\]

\end{example}

\vspace{-0.5cm}
\subsection*{A logical perspective on the \'Syena controversy}
\vspace{-0.08cm}
In {\Mimamsa} literature many explanations of the reasons why the sentences A. and B. in Sec. \ref{sec:extraction} are not contradictory have been proposed. 
We show that the $\MDL$ solution matches the one of Prabh\={a}kara, one of the chief {\Mimamsa} authors, and makes it formally 
meaningful. 

Consider the sequent in the proof of Thm.~\ref{thm:syena}. Since it is not derivable in $\GDL$, using 
Algorithm~\ref{alg:proof-search} we can construct a model for the formula 
\begin{equation}
\label{form:negsyseq}
 \Box (\harmenemy \to \harm) \land  \Box (\syena\to\harmenemy) \land \Box \Obl (\neg \harm / \top) \land \Box \Obl (\syena / \desharm)
\end{equation}
However, to make the solution clearer, we define below a simpler model $\mathfrak{M}_{0} = ( W_{0}, R_{0}, \eta_{0}), \sigma_{0}$ based on Vedic concepts.
The domain $W_{0}$ is $\lbrace w_{i} \mid 1 \leq i \leq 8\rbrace$, represented in Fig. \ref{fig:syenamodel} by circles. The accessibility relation $R_{0}$ is universal, i.e. for any $1 \leq i,j \leq 8$ it holds that $R_{0} ( w_{i}, w_{j})$; it is not represented in the figure for better readability. The map $\eta_{0}$ is such that $\eta_{0} (w_{i}) = \lbrace ( X ,  W _{0}) \; \mid \; X \subseteq W_{0}, \; \lbrace w_{1}, w_{5} \rbrace \subseteq X \rbrace  \; \bigcup \; \lbrace ( Y ,  \lbrace w_{5}, w_{6}, w_{7}, w_{8} \rbrace  ) \; \mid \; Y \subseteq W_{0}, \; \lbrace w_{4}, w_{8} \rbrace \subseteq Y \rbrace$. The figure represents only the elements of the neighbourhood of $w_{1}$ that are relevant to the valuation of our deontic statements. Each element corresponds to a kind of arrow: solid arrows for the statement about \'Syena and dashed ones for the obligation not to harm anyone. An element of the neighbourhood is a pair of sets of states, to represent it we draw an arrow from each state belonging to the second element of the pair to each one belonging to the first element of the pair. 
The function $\sigma_{0}$ is the valuation of the model and it is such that $\sigma_{0} (w_{1}) = \emptyset $; $\sigma_{0} (w_{2}) =  \lbrace \harm \rbrace $; $\sigma_{0} (w_{3}) =  \lbrace \harm ,$  $\harmenemy \rbrace $; $\sigma_{0} (w_{4}) =  \lbrace \harm ,$  $\harmenemy ,$ $\syena \rbrace $; $\sigma_{0} (w_{5}) =  \lbrace \desharm \rbrace $; $\sigma_{0} (w_{6}) =  \lbrace \harm ,$ $\desharm \rbrace $; $\sigma_{0} (w_{7}) =  \lbrace \harm ,$ $\harmenemy ,$ $\desharm \rbrace $; and $\sigma_{0} (w_{8}) =  \lbrace \harm ,$  $\harmenemy ,$ $\syena ,$  $\desharm \rbrace $. Clearly $\mathfrak{M}_{0}$ satisfies all the requirements stated in Def.~\ref{def:m-frame}.
\begin{figure}[t]

\tikzstyle{crc}=[circle, minimum size=5mm, inner sep=0pt, draw]

\begin{tikzpicture}[node distance=1.5cm,auto,>=latex', scale=0.5]
    \node [crc] (1) [label=180: {}] at (0,0) {$w_{1}$};
    \node [crc] (2) [label=0: {$ \harm$}] at (2,-1) {$w_{2}$};
    \node [crc] (3) [label=0: {$\harm$, $ \harmenemy$}] at (3,-3) {$w_{3}$};
    \node [crc] (4) [label=0: {$\harm$, $\harmenemy$, $\syena$}] at (2,-5) {$w_{4}$};
    \node [crc] (5) [label=290: {$\desharm$}] at (0,-6) {$w_{5}$};
    \node [crc] (6) [label=180: {$\harm$, $\desharm$}] at (-2,-5) {$w_{6}$};
    \node [crc] (7) [label=180: {$\harm$, $\harmenemy$, $\desharm$}] at (-3,-3) {$w_{7}$};
    \node [crc] (8) [label=180: {$\harm$, $\harmenemy$, $\syena$, $\desharm$}] at (-2,-1) {$w_{8}$};
    \foreach \i in {2,3,4,6, 8}
            \path[->] (\i) edge [] node {} (1);
    
    \path[->] (5) edge [bend right=5] node {} (1);

    \path[->] (7) edge [bend right=5] node {} (1);

    \path[->] (1) edge [loop above] node {} (1);

    \path[->] (5) edge [loop below] node {} (5);

    \foreach \i in {4, 8}
            \path[->] (\i) edge [bend right=5] node {} (5);

    \path[->] (1) edge [bend right=5] node {} (5);

    \path[->] (3) edge [bend right=5] node {} (5);
    
    \foreach \i in {2,6,7}
            \path[->] (\i) edge [] node {} (5);     

    \foreach \i in {6,7,8}
            \path[->] (\i) edge [dashed] node {} (4);

    \path[->] (5) edge [dashed, bend right=5] node {} (4);

    \foreach \i in {6,7}
            \path[->] (\i) edge [dashed] node {} (8);

    \path[->] (5) edge [dashed, bend right=5] node {} (8);

    \path[->] (8) edge [dashed, loop above] node {} (8);
\end{tikzpicture}
\caption{The model $\mathfrak{M}_{0}$ for the \'Syena controversy}
\label{fig:syenamodel}
\end{figure}
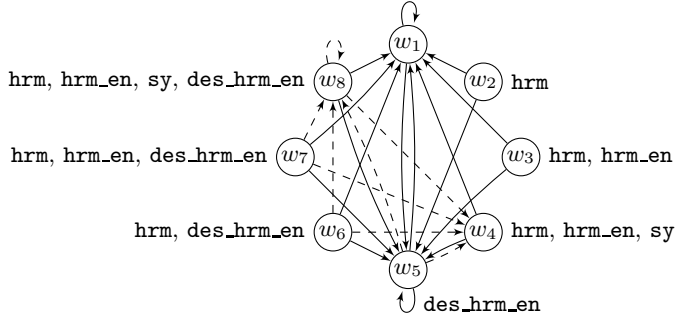

The definition of $\mathfrak{M}_{0}$ is based on \textit{adhik\={a}ra} (\cite{Freschi2012}, pp.147-155), a central concept in Prabh\={a}kara's analysis of the Vedas, which identifies the addressee of a prescription through their desire for the results. In the prescription about the \'Syena sacrifice, the \textit{adhik\={a}ra} corresponds to the desire to harm an enemy; the results correspond to the fact that an enemy is harmed through the performance of \'Syena, and, more generally, to the fact that someone is harmed. Some combinations of these facts are impossible if we need to satisfy $ \Box (\harmenemy \to \harm)$ and $ \Box (\syena\to\harmenemy)$, thus all the possibilities are the eight states in the model.
The accessibility relation accounts for the possible changes of subject's condition. The neighbourhood of a state encodes the obligations holding for that state, and given that these obligations are the same for each state, the neighbourhood is the same too. Thus the arrows show the changes of condition promoted by the obligations.

We show now that the formula \eqref{form:negsyseq} is true in the state $w_{1}$.
First, all its conjuncts without deontic operators are true in all states.
Secondly, the formula $\Box \Obl (\neg  \harm / \top)$ is true in $w_{1}$ if
$(\trs{\neg \harm}_{\mathfrak{M}_{0}} \cap R_{0}[s], \trs{\top}_{\mathfrak{M}_{0}} \cap R_{0}[s]) \in \eta_{0}(s) $ holds for all $s$ such that $R_{0}(w_{1}, s)$. 
Given that $ (  \lbrace w_{1} , w_{5} \rbrace  ,  W_{0} )$ belongs to $\eta_{0}(s)$ for all $s \in W_{0}$, the formula $\Obl(\neg  \harm / \top)$ is true in all states. For the formula $\Box \Obl (\syena / \desharm)$ the valuation is similar.
Hence $\mathfrak{M}_{0}$ is a model of \eqref{form:negsyseq} and, by Thm. \ref{thm:strong-compl} and \ref{thm:deduction}, this provides a semantic proof of Thm. \ref{thm:syena}.

Among the different solutions for the \'Syena controversy, the model $\mathfrak{M}_{0}$ matches Prabh\={a}kara's one
which can be summarised in his statement:
``\textit{A prescription regards what has to be done. But it does not say that
it has to be done}'' (\emph{B\d{r}hat\={\i}} I, p. 38, l. 8f).
Indeed in state $w_{1}$ 
no conflicting prescriptions are applicable and all obligations are fulfilled. We call this a \textit{Vedic state}. The existence of such a state shows that an agent can find a way not to transgress any prescription, and that the Vedic prescriptions do not imply that the \'Syena sacrifice has necessarily to be done.
Our model also explains Prabh\={a}kara's claim that \textit{the Vedas do not impel one to perform the malevolent sacrifice
\'Syena, they only say that it is obligatory}, which was wrongly considered meaningless e.g.
in \cite{Stcherbatsky1926}.

\begin{remark}
\label{AI}
Our analysis highlights that Vedic prescriptions are ``instructions to attain desired outcomes'' rather than absolute imperatives.
A \textit{Vedic state} provides a way not to transgress any obligation, but at the same time there are norms, e.g., the one about \'Syena, for those who intend to transgress some obligations, but nonetheless do not want to altogether reject the Vedic principles.
This is explicit in another {\Mimamsa} author, Ve\.{n}ka\d{t}an\={a}tha, who claims that the \'{S}yena is the best way to kill one's enemy if one is determined to transgress the general prescription not to perform violence.
This feature suggests a possible use of suitable extensions of $\bMDL$ to reason about machine ethics, where indeed choices between actions 
that should be avoided often arise. 
Consider a self-driving vehicle that has no choice but to harm some people. There is no perfect solution
but, nevertheless, the system should be able to provide instructions that promote imperfect outcomes in order to avoid the worst-case scenario.
\end{remark}

\section{Conclusions and Future Work}

We defined a novel deontic logic justified by principles elaborated
by {\Mimamsa} authors over the last 2,500 years,
and used its proof theory and semantics to analyse a notoriously
challenging example.
The fruits of this synergy of Logic and Indology can be gathered from both sides: 
The vast body of knowledge constituted by {\Mimamsa} texts can provide interesting new stimuli for the
logic community, and at the same time
our methods can lead to 
new tools for the analysis of philosophical and sacred texts.
Our investigation also raises a number of further research directions,
such as 
(i) a formal analysis of the concept of prohibition as discussed by {\Mimamsa} authors. Moreover, (ii) 
among the about 200 considered\footnote{Not all {\Mimamsa} metarules have been translated from Sanskrit so far, see \cite{Atinerpaper}.}  \emph{ny\={a}ya}s (50 of which were
on deontic principles), 
some hinted at the need for extending $\bMDL$ in
various directions: e.g., the principle ``\textit{the
agent of a duty needs to be the one identified by a
given prescription}'' (PMS 6.1.1--3) seems to
require 
first-order quantification; some metarules that distinguish between
different repetitions of the same action suggest the introduction of
{\it temporal operators}; finally the fact that \'{S}Bh 1.1.1 asserts
that the Vedas prevail over other authoritative texts suggests the
need of a system to manage conflicts among different authorities, a
feature 
also important 
for reasoning about ethical machines \cite{Vardi}.
Finally,
(iii)~while the metarules considered for $\bMDL$ are common to the
{\Mimamsa} school, there are additional principles 
employed only by specific authors. 
Their identification and formalisation might shed light on the strength of the different interpretations of 
various {\Mimamsa} authors and, e.g., help arguing for 
the conjecture that Kum\={a}rila's interpretation
is more explicative than 
Ma\d{n}\d{d}ana's.

\bibliographystyle{splncs03}


\end{document}